\documentclass[journal,10pt]{IEEEtran}
\usepackage{amsfonts}
\usepackage{times}
\usepackage{graphicx}
\usepackage{latexsym}
\usepackage{dsfont}
\usepackage{amssymb}
\usepackage{amsmath}
\usepackage{cite}
\usepackage{verbatim}
\usepackage{subfigure}
\newtheorem{theorem}{Theorem}

\newtheorem{algorithm}[theorem]{Algorithm}

\newtheorem{corollary}[theorem]{Corollary}

\newtheorem{proposition}[theorem]{Proposition}

\newenvironment{proof}{ \textbf{Proof:} }{ \hfill $\Box$}

\newcommand{\figref}[1]{{Fig.}~\ref{#1}}

\def\bb0{{\mathbb{0}}}

\def\ba{{\mathbf{a}}}
\def\bb{{\mathbf{b}}}

\def\bh{{\mathbf{h}}}

\def\bn{{\mathbf{n}}}

\def\bp{{\mathbf{p}}}

\def\br{{\mathbf{r}}}
\def\bs{{\mathbf{s}}}

\def\bw{{\mathbf{w}}}

\def\by{{\mathbf{y}}}

\def\b0{{\mathbf{0}}}

\def\bA{{\mathbf{A}}}

\def\bH{{\mathbf{H}}}
\def\bI{{\mathbf{I}}}

\def\bS{{\mathbf{S}}}

\def\bW{{\mathbf{W}}}

\def\bbE{{\mathbb{E}}}

\def\cA{\mathcal{A}}
\def\cB{\mathcal{B}}

\def\cS{\mathcal{S}}

\def\cU{\mathcal{U}}

\def\sf0{{\mathsf{0}}}

\newcommand{\sref}[1]{{Section}~\ref{#1}}

\usepackage{epstopdf}
\usepackage{enumerate}
\usepackage{algorithmicx}
\usepackage{algorithm}
\usepackage{amsmath}
\usepackage[noend]{algpseudocode}
\usepackage{float}
\usepackage{color}
\usepackage{makeidx}
\usepackage{bbm}

\allowdisplaybreaks

\begin{document}
	
\title{Massive MIMO Combining with Switches}
\author{Ahmed Alkhateeb, Young-Han Nam, Jianzhong (Charlie) Zhang, and Robert W. Heath, Jr.\thanks{Ahmed Alkhateeb and Robert W. Heath Jr. are with The University of Texas at Austin (Email: aalkhateeb, rheath@utexas.edu). Young-Han Nam and Jianzhong (Charlie) Zhang are with Samsung Research America-Dallas (Email: younghan.n, jianzhong.z@samsung.com).} \thanks{This work was done while the first author was with Samsung Research America-Dallas. The Authors at the University of Texas at Austin are supported in part by the National Science Foundation under Grant No. 1319556.}}

\maketitle

\begin{abstract}
	Massive multiple-input multiple-output (MIMO) is expected to play a central role in future wireless systems. The deployment of large antenna arrays at the base station and the mobile users offers  multiplexing and beamforming gains that boost system spectral efficiency. Unfortunately, the high cost and power consumption of components like analog-to-digital converters makes assigning an RF chain per antenna and applying typical fully digital precoding/combining solutions difficult. In this paper, a novel architecture for massive MIMO receivers, consisting of arrays of switches and constant (non-tunable) phase shifters, is proposed. This architecture applies a quasi-coherent combining in the RF domain to reduce the number of required RF chains. An algorithm that designs the RF combining for this architecture is developed and analyzed. Results show that the proposed massive MIMO combining model can achieve a comparable performance to the fully-digital receiver architecture in single-user and multi-user massive MIMO setups.
\end{abstract}


\section{Introduction} \label{sec:Intro}
Massive MIMO offers large multiplexing and array gain, making it an attractive technology for next generation cellular systems \cite{Boccardi2014}. The hardware and power requirements of the typical MIMO systems that allocate a complete RF chain per antenna become a limiting factor for the practical implementation of these systems. This paper proposes a new architecture for massive MIMO receivers that requires a small number of RF chains and low-power RF hardware components, while performing close to fully-digital solutions.

To reduce power consumption, hybrid analog/digital precoding architectures, which use a small number of RF chains and divide the precoding/combining processing between RF and baseband domains, have been proposed for large-scale MIMO systems \cite{ElAyach2014,Alkhateeb2014b,Zeng2014,Brady2013,Mendez-Rial2015}. The RF precoding circuit can be implemented using networks of variable phase shifters \cite{ElAyach2014,Alkhateeb2014b}, or lens antennas \cite{Zeng2014,Brady2013}. To further reduce power consumption, a hybrid architecture that uses switches instead of phase shifters has been proposed in \cite{Mendez-Rial2015}. The architecture in \cite{Mendez-Rial2015}, though, assumes that each RF chain connects only to one antenna and relies on antenna subset selection to design the RF precoders. This reduces the array gain since fewer antennas are active compared to other hybrid precoding solutions \cite{ElAyach2014,Alkhateeb2014b}.

In this paper, we propose a new architecture for massive MIMO receivers based on RF antenna switches and constant (non-tunable and pre-designed for the frequency band of interest) phase shifters, which are attractive components due to their low power consumption \cite{Mendez-Rial2015,Ahn2009}. Further, as the proposed architecture applies a quasi-coherent combining in the RF domain, it requires much smaller number of RF chains compared with the number of antennas. For this architecture, we develop a low-complexity algorithm that designs the combining matrix in polynomial time in terms of the number of antennas. The performance of the proposed architecture and developed algorithm is analyzed for both single-user and multi-user massive MIMO systems. Although the analytical results in this letter are derived for the proposed architecture, they can also be applied to evaluate the performance of conventional hybrid precoding architectures with quantized phase shifters. Numerical simulations verify the analytical results and show that the proposed architecture achieves a comparable performance to fully-digital solutions and phase-shifters based hybrid architectures.

\textbf{Notation:} we use the following notation: $\bA$ is a matrix, $\ba$ is a vector, $a$ is a scalar, and $\cA$ is a set. $\bA^\mathrm{T}$, $\bA^*$ are its transpose and conjugate transpose, respectively. $\ba_n$ is $n$th element of $\ba$. $[\bA]_{n,m}$ is the element at the $n$th row and $m$th column of $\bA$.

\section{System Model}\label{sec:Model}
Consider the massive MIMO system model depicted in \figref{fig:Model}, where a massive MIMO basestation (BS) receiver with $N$ antennas is serving $U$ users in an uplink massive MIMO transmission.  The receiver is assumed to employ a number of RF chains $N_\mathrm{RF} \geq U$. Further, for simplicity of exposition, we assume that the BS uses only $U$ (out of the $N_\mathrm{RF}$) RF chains to simultaneously serve the $U$ users. Each RF chain is connected to all the antennas via a network of switches and constant phase shifters.

In this paper, we focus on uplink transmission. Let $\bH$ denote the $N \times U$ channel matrix that gathers the channel vectors of the $U$ users as $\bH=\left[\bh_1, ..., \bh_U \right]$, and let $\bs$ represent the  transmitted symbols from the $U$ users, such that $\bbE\left[s_u s_u^*\right] = {1}$. As this paper focuses on the achievable performance of the proposed transceiver architecture in \figref{fig:Model} compared with fully-digital solutions, we assume for simplicity of exposition that perfect power control has already been performed, canceling the large-scale fading effects, such that $P$  represents the average received signal power from all users.  Further, we adopt a narrowband massive MIMO channel model in this paper where an IID channel matrix is assumed, with $\left[\bH\right]_{n,u} \sim \mathcal{CN}\left(0,1\right)$. This assumption is made to simplify the analysis. The proposed architecture and algorithm in this paper, though, can be applied to more  general channel statistics. If $\bn \sim \mathcal{CN}\left(\boldsymbol{0}, \sigma^2 \bI\right)$ represents the noise vector, then the discrete-time received signal is
\begin{equation}
\br=\sqrt{P} \bH \bs+ \bn.
\label{eq:rec_signal}
\end{equation}
At the receiver, a combining vector  $\bw_u, u=1, ..., U$ is used to combine the signal from the $u$th user. Defining $\bW=\left[\bw_1, ..., \bw_U\right]$, the received signal after processing is
\begin{equation}
\by=\sqrt{P} \bW^* \bH \bs+ \bW^* \bn.
\label{eq:signal_transmitted}
\end{equation}
As the combining matrix $\bW$ is implemented with switches and constant phase shifters, the elements of this matrix are subject to a constant modulus constraint. Further, as only $N_\mathrm{Q}$ constant phase shifts are possible, the combining vector of each user $u$ can be written as $\bw_u=\bS_u \bp$, where $\bp=[1, e^{-j\frac{2 \pi}{N_\mathrm{Q}}}, ..., e^{-j\frac{(N_\mathrm{Q}-1)2 \pi}{N_\mathrm{Q}}}{\vphantom{AA}]}^\mathrm{T}$, and $\bS_u \in \cB^{N \times N_\mathrm{Q}}$ is a binary switching matrix. To ensure that each antenna is switched to only one phase shifter per RF chain, the switching matrix $\bS_u$ is subject to the constraint $\sum_{q=1}^{N_\mathrm{Q}} \left[\bS_u\right]_{n,q}=1$. Note that the model in \eqref{eq:signal_transmitted} and the performance analysis in \sref{sec:Performance} do not assume any baseband combining processing. Additional combining, though, can be done in the baseband to further improve the performance as will be shown in \sref{sec:Results}.

\begin{figure} [t]
\centerline{\includegraphics[width=\linewidth]{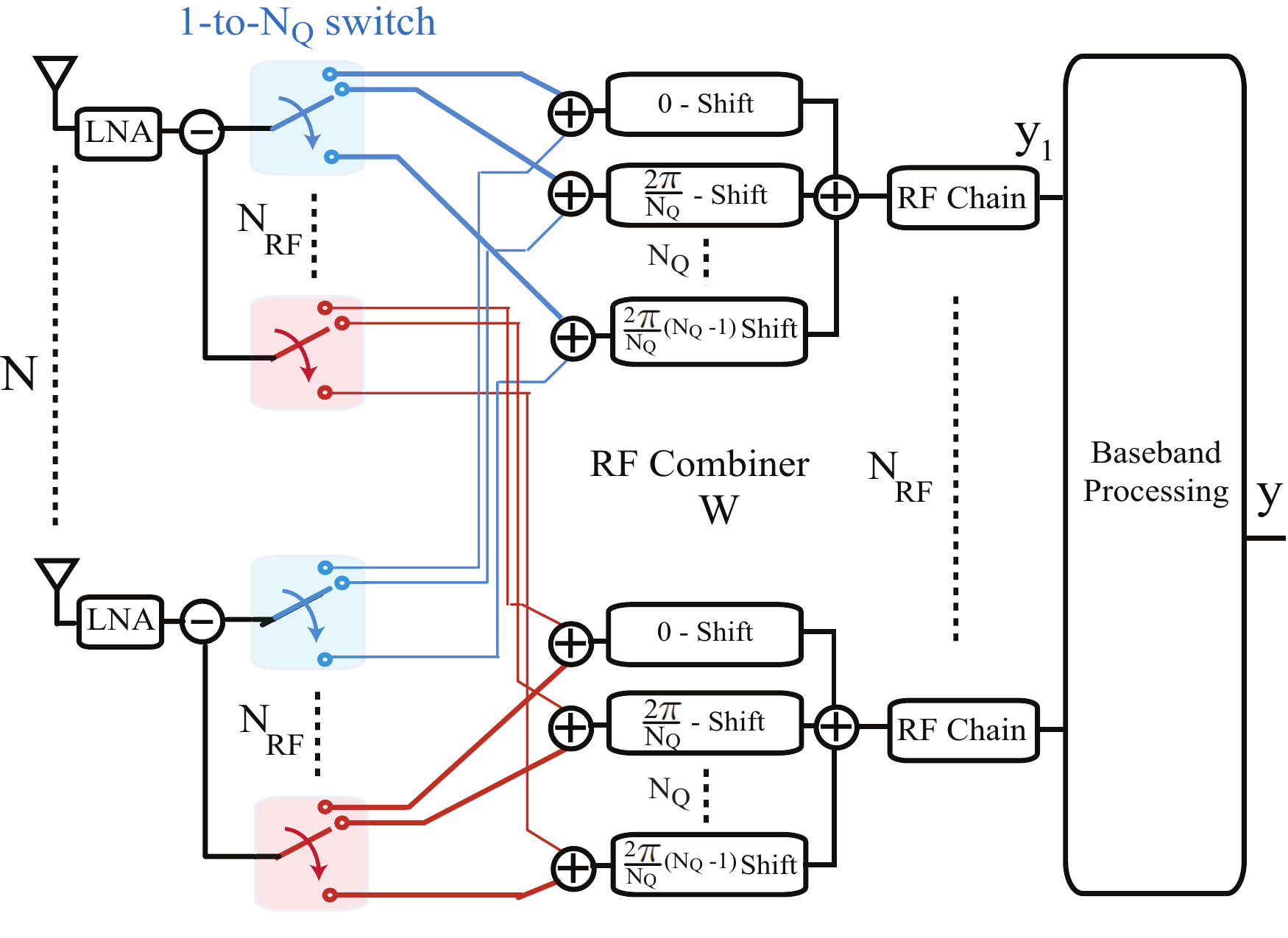}
}
\caption{A block diagram of the proposed switches-based massive MIMO receiver with $N$ antennas, $N_\mathrm{RF}$ RF chains, and $N_\mathrm{Q}$ constant (non-tunable) phase shifters per RF chain.}
\label{fig:Model}
\end{figure}

\begin{algorithm} [!t]                     
	\caption{Quasi-Coherent Switch Combining}          
	\label{alg:SCAP}                           
	\begin{algorithmic} 
		\State \textbf{Initialization:} $\phi_q=\frac{(2 q-1) \pi }{N_\mathrm{Q}}, q=1, ..., N_\mathrm{Q}$
		\For {$u, u=1, ..., U$}
		\State $\cS_q^u=\phi$, $\bS_u=\boldsymbol{0}$
		\For {$n, n=1,...,N$ }
		\State $q^\star = \displaystyle{\arg \hspace{-12pt} \min_{q=1, ..., N_\mathrm{Q}}} \left(\angle{h_{u,n}}\right)_{\left[0,2 \pi\right]} - \phi_q$
		\State $ \cS_{q^\star}^u = \cS_{q^\star}^u \cup \left\{n\right\}$, $\left[\bS_u\right]_{n,q^\star}=1$
		\EndFor
		\EndFor	
	\end{algorithmic}
\end{algorithm}
\section{Quasi-Coherent Combining Algorithm} \label{sec:Proposal}
Considering the system model in Section \ref{sec:Model}, the signal-to-interference-plus-noise ratio ($\mathsf{SINR}_u$) of user $u$ is
\begin{equation}
\mathsf{SINR}_u=\frac{P \left|\bw_u^* \bh_u\right|^2}{P \sum_{\ell \neq u}^{U} \left|\bw_u^* \bh_\ell\right|^2+\sigma^2 \left\|\bw_u\right\|^2}. \label{eq:SINR}
\end{equation}

The objective of our design for the combining matrix $\bW$ is to maximize the system sum rate. Taking into consideration the hardware constraints on the combining matrix implementation, this design problem can be formulated as
\begin{align}
\left\{\bS_u^\star\right\}_{u=1}^U  = & \arg\max \ \ \sum_{u=1}^{U}  \log_2\left(1+\mathsf{SINR}_u\right) \nonumber\\
& \mathrm{s.t.} \hspace{23 pt}  \left[\bS_u\right]_{n,q} \in \left\{0,1\right\}, \ \ \forall u,n,q, \label{eq:Opt}\\
& \hspace{35pt}  \sum_{q=1}^{N_\mathrm{Q}} \left[\bS_u\right]_{n,q}  = 1, \nonumber
\end{align}
where the first constraint is due to the use of switches and the second constraint is due to the restriction imposed on each antenna to be connected to exactly one switch. The problem in \eqref{eq:Opt} represents a non-coherent combining problem with combining vectors $\bw_u=\bS_u \bp, u=1, ..., U$ taken form a quantized space. Finding the optimal solution requires an exhaustive search over all $\left(N_\mathrm{Q}\right)^{N_\mathrm{RF} N} $ possible solutions. For large antenna systems, the complexity will likely be too high to directly implement the exhaustive search. Targeting a low-complexity solution, we propose Algorithm \ref{alg:SCAP} that greedily designs the switch selection matrices $\left\{\bS_u\right\}_{u=1}^{U}$ in polynomial time in terms of the number of antennas.

Algorithm \ref{alg:SCAP} can be summarized as follows. For each user $u$, the angle of the $n$th channel element $\left[\bh_u\right]_{n}$, $n=1, 2, ..., N$, is compared with the quantization angles $\frac{(2 q-1) \pi }{ N_\mathrm{Q}}, q=1, ..., N_\mathrm{Q}$. The $n$th antenna switch then selects the constant phase shifter from $\bp$ that corresponds to the closest quantization angle. Physically, this means that Algorithm 1 partitions the entries of each channel vector $\bh_u$ into $N_\mathrm{Q}$ sets based on the phases of the complex vectors representing the entries of $\bh_u$. This partitioning is done such that the $q$th set, $q=1, 2, ..., N_\mathrm{Q}$, contains the entries whose phases are in $[\left(q-1\right)\frac{2 \pi}{N_\mathrm{Q}}, q\frac{2 \pi}{N_\mathrm{Q}}]$. This is implemented using the switches as each set of them selects one of the $N_\mathrm{Q}$ phase shifters. Each set of channel entries are then rotated in the complex plane to the first phase sector, $\left[0, \frac{2 \pi}{N_\mathrm{Q}}\right]$, using the phase shifters before combining the channel sets together. Therefore, Algorithm \ref{alg:SCAP} combines the channel elements quasi-coherently as the phase differences between all the channel entries after switching and rotation are less than $\frac{2 \pi}{N_\mathrm{Q}}$.  
Despite its low-complexity, Algorithm \ref{alg:SCAP} can achieve good performance compared with the fully-digital high-complexity solution, as will be shown in the next section.

\section{Performance Analysis} \label{sec:Performance}
In this section, we analyze the performance of the proposed algorithm and receiver architecture in massive MIMO systems.

\subsection{Single User Systems}
\begingroup
\allowdisplaybreaks
\begin{theorem} For single-user systems, if $\mathsf{SNR}_u^\mathrm{MRC}$ denotes the achievable SNR with maximum ratio combining (MRC) and fully-digital receiver, and $\mathsf{SNR}_u^\mathrm{SC}$ denotes the achievable SNR with Algorithm \ref{alg:SCAP} and the receiver in \figref{fig:Model}, then
	\begin{equation}
	\lim_{N \to \infty} \frac{\mathsf{SNR}_u^\mathrm{SC}}{\mathsf{SNR}_u^\mathrm{MRC}} \xrightarrow{\text{a.s.}} \frac{N_\mathrm{Q}^2}{4 \pi} \sin^2\left(\frac{\pi}{N_\mathrm{Q}}\right).
	\end{equation}
	\label{thm:SNR}
\end{theorem}
\begin{proof} See Appendix \ref{app:SNR}
\end{proof}

This result means that around $65.6\%$ of the MRC SNR gain, with fully-digital hardware, can be achieved with the proposed algorithm and receiver architecture using only $N_\mathrm{Q}=4$ constant phase shifters. This result is also verified by numerical simulation as will be discussed in \sref{sec:Results}.
\endgroup
\subsection{Multi-User Systems}
The previous results can also be extended to multi-user systems, as captured by the following proposition.
\begin{proposition} The achievable SINR of Algorithm \ref{alg:SCAP} with the massive MIMO receiver in \sref{sec:Model} converges almost surely to $\frac{N}{U} \frac{N_\mathrm{Q}^2}{4 \pi} \sin^2\left(\frac{\pi}{N_\mathrm{Q}}\right)$  as $N, U \to \infty$ and $\frac{N}{U}$ is a constant ratio, i.e.,
\begin{equation}
	\lim_{\substack{N,U \to \infty \\ \frac{U}{N}=\text{constant}}} \mathsf{SINR}_u^\mathrm{SC} \xrightarrow{\text{a.s.}} \frac{N}{U} \frac{N_\mathrm{Q}^2}{4 \pi} \sin^2\left(\frac{\pi}{N_\mathrm{Q}}\right). 
\end{equation}
\label{prop:SINR}
\end{proposition}

\begin{figure} [t]
	\centerline{
		\includegraphics[width=1.02\linewidth]{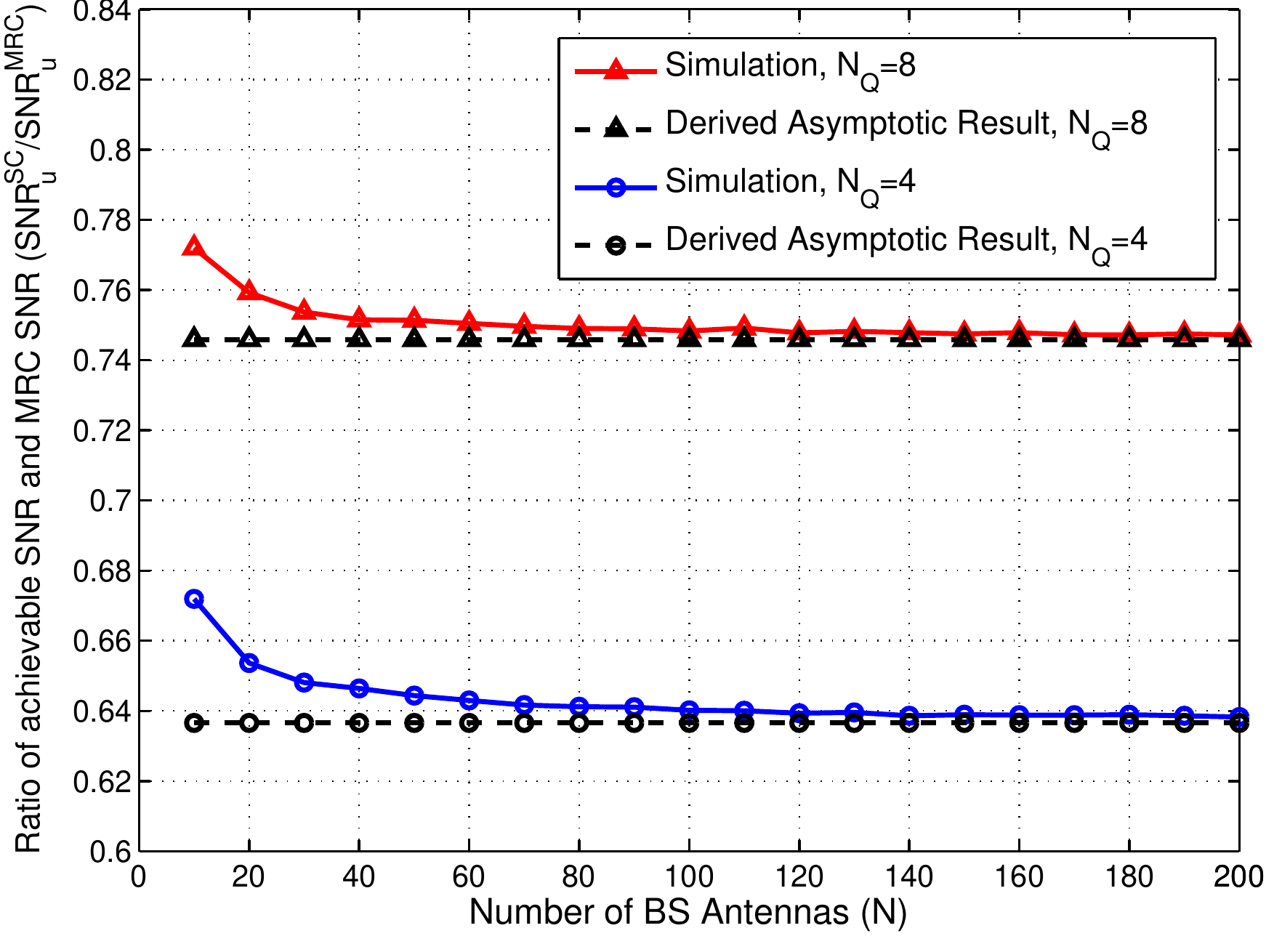}
	}
	\caption{Ratio between the achievable SNR of the proposed algorithm and receiver architecture and the MRC SNR converges to the derived analytical limit in Theorem \ref{thm:SNR} as $N \to \infty$.}
	\label{fig:BF}
\end{figure}

\begin{proof} 
	The $\mathsf{SINR}_u$ of user $u$ can be written as in \eqref{eq:SINR}. To prove that $\mathsf{SINR}$ converges almost surely to $\frac{N}{U} \frac{N_\mathrm{Q}^2}{4 \pi} \sin^2\left(\frac{\pi}{N_\mathrm{Q}}\right)$, it is sufficient to prove the following
	\begin{align}
	\lim_{N \to \infty}\frac{1}{N^2}\left|\bw_u^* \bh_u\right|^2 & \xrightarrow{\text{a.s.}} \frac{N_\mathrm{Q}^2}{4 \pi} \sin^2\left(\frac{\pi}{N_\mathrm{Q}}\right) \label{eq:Lim1}\\
	\lim_{\substack{N,U \to \infty \\ \frac{U}{N}=\text{constant}}} \frac{1}{N^2}\sum_{\ell \neq u}^{U} \left|\bw_u^* \bh_\ell\right|^2 &  \xrightarrow{\text{a.s.}} \frac{U}{N}   \label{eq:Lim2}
	\end{align}
	The first limit in \eqref{eq:Lim1} is proved in Theorem \ref{thm:SNR}. To prove the limit in \eqref{eq:Lim2}, we first note that given the combining vectors design in Algorithm \ref{alg:SCAP}, we have $\left\|\bw_u^* \bh_\ell\right\|^2 \sim \text{Exp}(N), \forall \ell \neq u$ (proof is omitted for space limitation). Then, by applying the law of large numbers, we reach \eqref{eq:Lim2}.
\end{proof}

Using Proposition \ref{prop:SINR},  we can characterize the achievable rate  with the proposed transceiver architecture as shown in the following corollary.

\begin{corollary} The achievable rate  Algorithm \ref{alg:SCAP} with the massive MIMO receiver in \sref{sec:Model} converges almost surely to $\log_2\left(1+\frac{N}{U} \frac{N_\mathrm{Q}^2}{4 \pi} \sin^2\left(\frac{\pi}{N_\mathrm{Q}}\right)\right)$  as $N, U \to \infty$ and $\frac{N}{U}$ is a constant ratio.
\label{cor:Rate}
\end{corollary}

The proof of Corollary \ref{cor:Rate} follows directly from the $\mathsf{SINR}$ result in Proposition \ref{prop:SINR}  by using the dominated convergence and the continuous mapping theorems \cite{Hoydis2013}.

\section{Simulation Results} \label{sec:Results}
To verify these results by numerical simulations, \figref{fig:BF} shows the ratio of the achievable SNR with the proposed architecture and combining algorithm over the optimal MRC SNR when only user $u$ is served. \figref{fig:BF} illustrates that this ratio approaches the derived $\frac{N_\mathrm{Q}^2}{4 \pi} \sin^2\left(\frac{\pi}{N_\mathrm{Q}}\right) $ approximation result in Theorem \ref{thm:SNR} as the number of receive antennas increases.

\begin{figure} [t]
	\centerline{
		\includegraphics[width=1.01\linewidth]{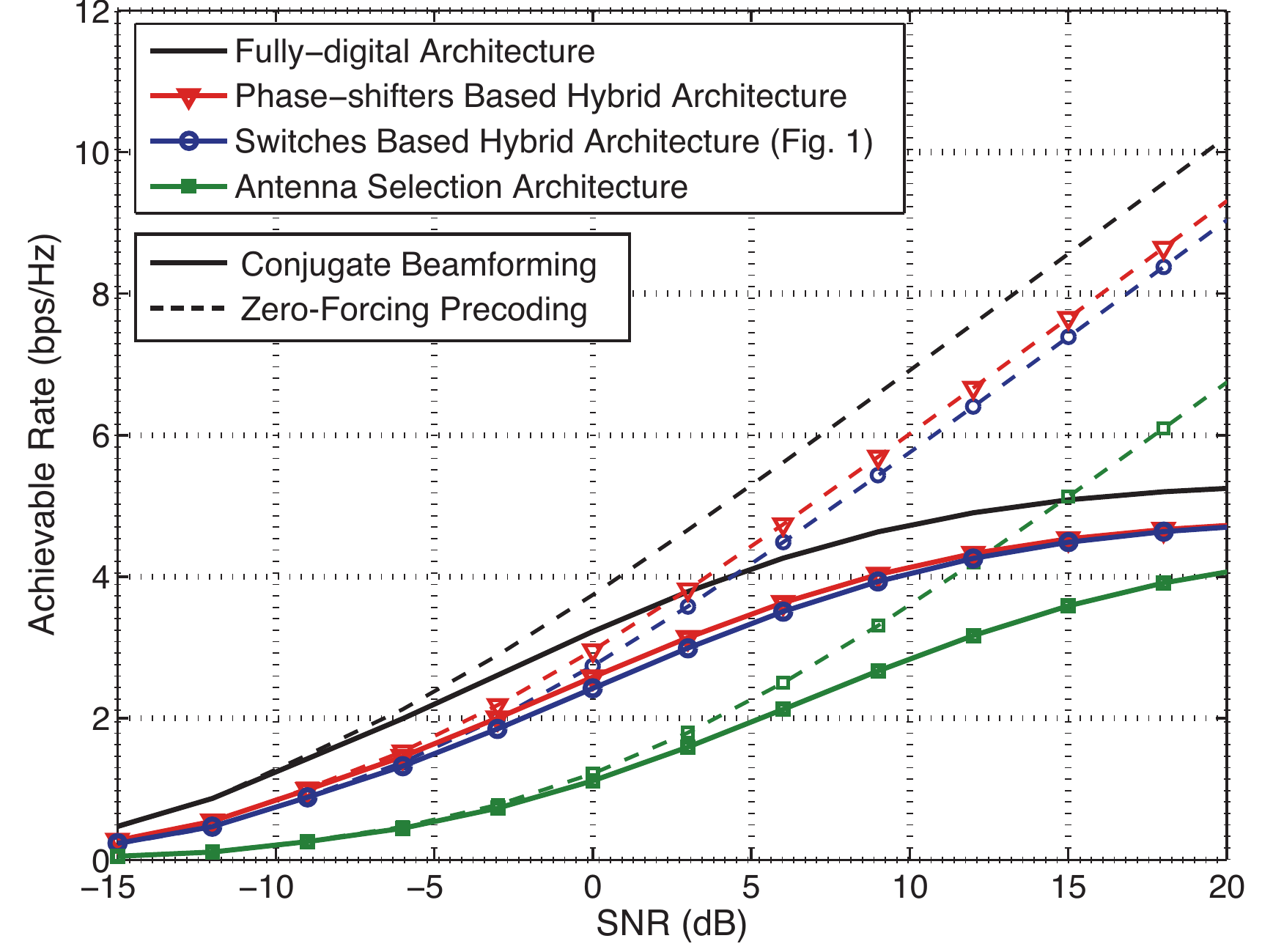}
	}
	\caption{The per-user achievable rate of the proposed receiver architecture is compared with the fully-digital, phase-shifers based hybrid precoding, and antenna-selection architectures for different combining solutions, assuming a system with $N=64$ antennas and  $U=3$ users.}
	\label{fig:Rate}
\end{figure}

In \figref{fig:Rate}, the achievable rate of the switches-based hybrid  architecture in \figref{fig:Model} is compared with the fully-digital, antenna selection, and phase-shifters based hybrid receivers in a multi-user setting with $N=64$ antennas, and $U=3$ users. Combining is assumed to be done with either matched-filtering or zero-forcing. For the switches-based and phase-shifters based hybrid architectures, we adopt the zero-forcing solution in \cite{Alkhateeb2014b}, in which analog combining vectors are designed per user (similar to that in \sref{sec:Proposal}) and zero-forcing is then done in the baseband based on the effective channels. For the antenna-selection solutions, we assume the receiver selects the best $U$ antennas that maximize the achievable rate subject to the matched-filter or zero-forcing solution via an exhaustive search over all antenna subset combinations. To account for the possible SNR degradation in the different receiver architectures, we penalized the SNR of each receiver with its noise figure (NF). We assumed receivers use LNA's of $22$ dB gain, and $5$ dB NF, and mixers of $12$ dB NF. The power dividers, combiners, and switches in the hybrid architectures are assumed to be implemented through cascaded stages using two-ports building blocks. The NF of each combining/dividing stage is $1$dB and each switching stage is $1.5$ dB. Phase shifters with $4$ dB NF are assumed to be accompanied by VGAs of $4$ dB NF to compensate for their gains. These NF values are collected from various references for RF components operating at $2-5$ GHz frequency.
With these values, the composite NF of the fully-digital and antenna-selection receivers is $5.1$ dB, while that of phase-shifters based and switched based hybrid architectures are $5.7$ dB, and $7.2$ dB, respectively. \figref{fig:Rate} shows that the performance of the switches-based architecture is close to the fully-digital solution with good gain over the antenna selection architecture thanks to the higher array gain. Finally, the architecture with switches is shown to have a small loss compared to that with phase shifters. In terms of power consumption, though, the main difference between those two hybrid architectures is that the one with phase shifters requires $N \times N_\mathrm{RF}=192$ controllable phase shifters, while that with switches requires only $N_\mathrm{RF} \times N_Q=12$ constant phase shifters. The other power consumption circuits like that activating the switches are similar in the two hybrid architectures. Note also that the receiver calibration techniques for the two architectures are expected to be similar, as they both can realize the same combining matrices. It is worth mentioning here that this initial study did not consider other factors like switches isolation and speed, whose impacts on the performance need to be investigated for a refined conclusion about the different architectures. 

\section{Conclusion}
In this paper, we proposed a receiver architecture for massive MIMO systems based on antenna switching. For this architecture, we developed and analyzed a quasi-coherent RF combining algorithm. Results showed that the proposed receiver can achieve at least a $\frac{N^2_\mathrm{Q}}{4 \pi} \sin^2\left(\frac{\pi}{N_\mathrm{Q}}\right)$ of the optimal maximum-ratio combining gain with $N_Q$ constant phase shifters. Results also illustrated that the achievable rate with the proposed receiver architecture is comparable to the fully-digital solutions in multi-user massive MIMO systems. For future research, it would be interesting to investigate the performance of the proposed architecture incorporating the impact of practical factors like switches imperfect isolation, switching speed, and more detailed insertion losses.

\appendices
\section{}\label{app:SNR}
\begin{proof}[Proof of Theorem \ref{thm:SNR}]
	Define $\rho=\frac{P}{\sigma^2}$, the SNR of the MRC with fully-digital receiver is $\mathsf{SNR}_u^{\text{MRC}}= \rho \sum_{n=1}^{N} |h_{u,n}|^2$.
	 
	 As the \textit{favorable} massive MIMO propagation conditions hold for our channel model \cite{Ashikhmin2012}, $\lim_{N \to \infty}\frac{1}{N} \mathsf{SNR}_u^{\text{MRC}} \xrightarrow{\text{a.s.}} \rho$.
	
	Now, considering Algorithm \ref{alg:SCAP} for the receiver in \figref{fig:Model}, the achieved SNR with $N_\mathrm{Q}$ phase shifters is $\mathsf{SNR}_u^\mathrm{SC}=\frac{\rho \left|\bw_u^* \bh_u\right|^2}{\left\|\bw_u\right\|^2}$.
	
	The desired signal power term can be expanded as
	\begin{align}
	\left|\bw_u^* \bh_u\right|^2&=\left|\sum_{n} w_{u,n} h_{u,n}\right|^2
	=\left|\sum_{q=1}^{N_\mathrm{Q}}{\sum_{n \in \cS_q} {h_{u,n} e^{- j \frac{(q-1) 2 \pi }{N_\mathrm{Q}}} }}\right|^2, \nonumber \\
	&\hspace{-40pt}= \left|\sum_{q=1}^{N_\mathrm{Q}}{\sum_{n \in \cS_q} {\left|h_{u,n}\right| \cos\left(\bar{\theta}^{q}_{u,n}\right) }}\right|^2 + \left| \sum_{q=1}^{N_\mathrm{Q}}{\sum_{n \in \cS_q} {\left|h_{u,n}\right| \sin\left(\bar{\theta}^{q}_{u,n}\right) }}\right|^2
	\end{align}
	where $\bar{\theta}^{q}_{u,n}=\theta_{u,n}-\frac{2 \pi (q-1)}{N_\mathrm{Q}}$. Now, we use the law of large numbers to get $\lim_{N \to \infty} \frac{1}{N^2} \left\|\bw_u^* \bh_u\right\|^2  \xrightarrow{\text{a.s.}}$
	\begin{align}
	&  \left(\bbE\left[\left|h_{u,n}\right| \cos(\theta_{u,n})\left| \theta_{u,n} \sim \cU \left(0, \frac{2 \pi}{N_\mathrm{Q}}  \right) \right. \right] \right)^2\nonumber\\
	&  +\left(\bbE\left[\left|h_{u,n}\right| \sin(\theta_{u,n})\left| \theta_{u,n} \sim \cU \left(0, \frac{2 \pi}{N_\mathrm{Q}} \right) \right. \right]\right)^2, 
	\label{eq:12}
	\end{align} 
	where this follows by noting that $\bar{\theta}_{u,n}^{q} \in \left[0, \frac{2 \pi}{N_\mathrm{Q}}\right], \forall n, u, q$. The almost sure convergence follows directly by noting that the massive MIMO favorable conditions hold for the adopted model with $|h_{u,n}|$ Rayleigh distributed and ${\theta}_{u,n}$  uniformly distributed \cite{Ashikhmin2012}. Leveraging the independence between $|h_{u,n}|$ and $\theta_{u,n}$, we reach $\lim_{N \to \infty} \frac{1}{N^2} \left\|\bw_u^* \bh_u\right\|^2  \xrightarrow{\text{a.s.}} $
	\begin{align}
	&   \left(\bbE\left[\left|h_{u,n}\right|  \right]\right)^2 \left(\bbE\left[\cos(\theta_{u,n})\left| \theta_{u,n} \sim \cU \left(0, \frac{2 \pi}{N_\mathrm{Q}} \right. \right) \right] \right)^2\nonumber \\ & +\left(\bbE\left[\left|h_{u,n}\right|  \right]\right)^2 \left(\bbE\left[\sin(\theta_{u,n})\left| \theta_{u,n} \sim \cU \left(0, \frac{2 \pi}{N_\mathrm{Q}} \right) \right. \right]\right)^2 \\
	& \stackrel{(a)}{=} \frac{N_\mathrm{Q}^2}{16 \pi} \left(  \sin^2\left(\frac{2 \pi}{N_\mathrm{Q}}\right) + \left(1-\cos\left(\frac{2 \pi}{N_\mathrm{Q}}\right)\right)^2 \right), \\
	& \stackrel{(b)}{=} \frac{N_\mathrm{Q}^2}{4 \pi}  \sin^2\left(\frac{ \pi}{N_\mathrm{Q}}\right),
	\label{eq:15}
	\end{align}
	where (a) follows from the distributions of $|h_{u,n}|$ and $\theta_{u,n}$, and (b) is by using the trigonometric identities. Finally, we note that $\left\|\bw_u\right\|^2=N$ given the adopted switching based receiver architecture. Hence, we reach 
	\begin{equation}
	\lim_{N \to \infty} \frac{1}{N}\mathsf{SNR}_u^\mathrm{SC} \xrightarrow{\text{a.s.}} \frac{N_\mathrm{Q}^2}{4 \pi} \sin^2\left(\frac{\pi}{N_\mathrm{Q}}\right) \rho,
	\end{equation}
	which completes the proof.
\end{proof}

\bibliographystyle{IEEEtran}

\begin{thebibliography}{1}
	\providecommand{\url}[1]{#1}
	\csname url@samestyle\endcsname
	\providecommand{\newblock}{\relax}
	\providecommand{\bibinfo}[2]{#2}
	\providecommand{\BIBentrySTDinterwordspacing}{\spaceskip=0pt\relax}
	\providecommand{\BIBentryALTinterwordstretchfactor}{4}
	\providecommand{\BIBentryALTinterwordspacing}{\spaceskip=\fontdimen2\font plus
		\BIBentryALTinterwordstretchfactor\fontdimen3\font minus
		\fontdimen4\font\relax}
	\providecommand{\BIBforeignlanguage}[2]{{%
			\expandafter\ifx\csname l@#1\endcsname\relax
			\typeout{** WARNING: IEEEtran.bst: No hyphenation pattern has been}%
			\typeout{** loaded for the language `#1'. Using the pattern for}%
			\typeout{** the default language instead.}%
			\else
			\language=\csname l@#1\endcsname
			\fi
			#2}}
	\providecommand{\BIBdecl}{\relax}
	\BIBdecl
	
	\bibitem{Boccardi2014}
	F.~Boccardi, R.~Heath, A.~Lozano, T.~Marzetta, and P.~Popovski, ``Five
	disruptive technology directions for {5G},'' \emph{IEEE Comm. Mag.}, vol.~52,
	no.~2, pp. 74--80, Feb. 2014.
	
	\bibitem{ElAyach2014}
	O.~El~Ayach, S.~Rajagopal, S.~Abu-Surra, Z.~Pi, and R.~Heath, ``Spatially
	sparse precoding in millimeter wave {MIMO} systems,'' \emph{IEEE Trans. on
		Wireless Commun.}, vol.~13, no.~3, pp. 1499--1513, Mar. 2014.
	
	\bibitem{Alkhateeb2014b}
	A.~Alkhateeb, G.~Leus, and R.~Heath, ``Limited feedback hybrid precoding for
	multi-user millimeter wave systems,'' \emph{IEEE Trans. on Wireless Commun.},
	vol.~14, no.~11, pp. 6481--6494, Nov. 2015.
	
	\bibitem{Zeng2014}
	Y.~Zeng, R.~Zhang, and Z.~N. Chen, ``Electromagnetic lens-focusing antenna
	enabled massive {MIMO}: Performance improvement and cost reduction,''
	\emph{IEEE Jour. on Select. Areas in Commun.}, vol.~32, no.~6, pp.
	1194--1206, June 2014.
	
	\bibitem{Brady2013}
	J.~Brady, N.~Behdad, and A.~Sayeed, ``Beamspace {MIMO} for millimeter-wave
	communications: System architecture, modeling, analysis, and measurements,''
	\emph{IEEE Trans. on Ant. and Propag.}, vol.~61, no.~7, pp. 3814--3827, July
	2013.
	
	\bibitem{Mendez-Rial2015}
	R.~M{\'e}ndez-Rial, C.~Rusu, A.~Alkhateeb, N.~Gonz{\'a}lez-Prelcic, and R.~W.
	Heath~Jr, ``Channel estimation and hybrid combining for mmwave: Phase
	shifters or switches?'' in \emph{Info. Theory and App. Workshop}, 2015.
	
	\bibitem{Ahn2009}
	M.~Ahn, H.-W. Kim, C.-H. Lee, and J.~Laskar, ``A 1.8-{GHz} 33-{dBm} {P}
	0.1-{dB} {CMOS} {T/R} switch using stacked {FETs} with feed-forward
	capacitors in a floated well structure,'' \emph{IEEE Transactions on
		Microwave Theory and Techniques}, vol.~57, no.~11, pp. 2661--2670, Nov. 2009.
	
	\bibitem{Hoydis2013}
	J.~Hoydis, S.~ten Brink, and M.~Debbah, ``Massive {MIMO} in the {UL/DL} of
	cellular networks: How many antennas do we need?'' \emph{IEEE Jour. on
		Select. Areas in Commun.}, vol.~31, no.~2, pp. 160--171, Feb. 2013.
	
	\bibitem{Ashikhmin2012}
	A.~Ashikhmin and T.~Marzetta, ``Pilot contamination precoding in multi-cell
	large scale antenna systems,'' in \emph{IEEE Int. Symp. on Info. Theo.}, July
	2012, pp. 1137--1141.
	
\end{thebibliography}

\end{document}